\documentclass[journal,a4paper]{IEEEtran}

\usepackage{amsmath}
\usepackage{amsthm}
\usepackage{cases}
\usepackage{amsfonts}
\usepackage{amssymb}
\usepackage{boxedminipage}
\usepackage{graphicx}
\usepackage[usenames]{color}
\usepackage{array,color}
\usepackage{colortbl}
\usepackage{extarrows}
\usepackage{bm}
\usepackage[numbers,sort&compress]{natbib}
\usepackage{epsfig}
\usepackage{subfigure}

\usepackage{stfloats}
\newcounter{TempEqCnt}

\newtheorem{remark}{Remark}
\newtheorem{theorem}{Theorem}
\newtheorem{lemma}{Lemma}

\newtheorem{definition}{Definition}

\begin{document}

\title{ Message Importance Measure and Its Application to Minority Subset Detection in Big Data}
\author{Pingyi~Fan, ~\IEEEmembership{Senior Member,~IEEE},  Yunquan~Dong,~Jiaxun~Lu,~Shanyun~Liu\\
State Key Laboratory on Microwave and Digital Communications\\
Tsinghua National Laboratory for Information Science and Technology \\
Department of Electronic Engineering, Tsinghua University, Beijing, P.R. China\\
E-mail:~fpy@tsinghua.edu.cn, dongyq08@mails.tsinghua.edu.cn
}
\maketitle

\vspace{-10mm}
\begin{abstract}
 Message importance measure (MIM) is an important index to describe the message importance in the scenario of big data.
   Similar to the Shannon Entropy and Renyi Entropy, MIM is required to characterize the uncertainty of a random process and some related statistical characteristics.
 Moreover, MIM also need to highlight the importance of those events with relatively small occurring probabilities, thereby is especially applicable to big data.
In this paper, we first define a parametric  MIM measure from the viewpoint of information theory and then investigate its properties.
    We also present a parameter selection principle that provides answers to the minority subsets detection problem in the statistical processing of big data.
\end{abstract}


\section{Introduction} \label{sec:1_intro}
In the big data era, the amount of data is growing steeply in a variety of areas.
    In addition to those applications that are benefited from big data, there are also cases in which only a small percent of data attracts people's interests.
For example,  in the national anti-terrorist system, only the flows of a small number of people and hazardous substance/physical agent need to be closely supervised \cite{antiterror15}.
    Also, for the synthetic ID detection \cite{synthesis10}, a few identities are artificially generated 
    for the purpose of committing financial frauds.
Embedded in a huge amount of data, the detection of these minority subsets, which is also known as the atypical event detection, becomes more and more challenging.

In the framework of rate-distortion theory, minority subset detection were investigated as a probabilistic clustering problem \cite{Ando07, Ando06, Crammer04,  Gupta06}.
    By classifying the events into a number of clusters under a certain optimal criteria (e.g., the minimum within-cluster distance, the minimum compressing distortion),  several clustering approaches were proposed.
In particular, the minority events can be recognized because their distribution contrast significantly with that of the majority of the dataset.
     Moreover, a graph-based rare category detection which recognizes atypical events using the global similarity matrix was proposed in \cite{He08}.
By further considering the time-evolving of graphs, a time-flexible rare category detection algorithm was presented in \cite{He15}.
    Although these algorithm are very efficient in their respective applications, it is noted that they were developed based on traditional information measures and frameworks, which were originally designed for the processing of typical events, i.e., those majority events.

As known, Shannon entropy \cite{Shannon48} and Renyi entropy \cite{Cover06, Erven2014} are two of the most fundamental measures in information theory and its applications.
    There also have been many successful attempts generalizing the two concepts.
Among the literature, the Kullback-Leibler divergence, f-divergence and Renyi divergence, as well as Fisher information, are the most popular ones \cite{Erven2014, Akaike1992}.
    Although these information measures have found their applications in communication theory, statistical parameter estimation, hypothesis testing or data analysis, they are no longer suitable for the minority subset detections in the big data scenarios.
In fact, these measures are focused on the encoding and signal processing of those typical sets of data.
    In order to facilitate the detection of those atypical sets of data other than conventional typical sets of data in the big data era,  new information measures are needed.

In this paper, we introduce a Message Importance Measure (MIM) that focus on those small-probability events.
    Before we proceed, let us review some of the major characteristics of Shannon entropy and Renyi Entropy first.

\subsection{Shannon Entropy and Renyi entropy}
In the finite alphabet case, for a given probability distribution $\boldsymbol{p}=(p_1,p_2,\cdots,p_n)$, the \textit{Shannon entropy} $H(\boldsymbol{p})$ is defined as
\begin{equation}
H(\boldsymbol{p})= -\sum_{i=1}^{n} p_i \log p_i,
\end{equation}
which measures the uncertainty, or the information that the distribution contains.

The \textit{Renyi entropy} $H_\alpha (\boldsymbol{p})$ of order $\alpha$ is defined as
\begin{equation}
H_\alpha (\boldsymbol{p})=\frac{1}{1-\alpha}\log \sum_{i=1}^n p_i^{\alpha},
\end{equation}
for $0<\alpha<\infty, \alpha\neq1$.
   Note that if we set $\alpha\rightarrow1$, the Renyi entropy converges to the Shannon entropy \cite{Cover06}.

In particular, both the Shannon entropy and the Renyi entropy have the following properties.
\begin{enumerate}
  \item They are non-negative;
  \item For the uniform distribution $\boldsymbol{u}=(1/n, 1/n, \cdots, 1/n)$, we have
            \begin{equation}
                H(\boldsymbol{u})=H_\alpha(\boldsymbol{u})=\log n.
            \end{equation}
  \item For any distribution $\boldsymbol{p}=(p_1,p_2,\cdots,p_n)$ without zero elements, i.e., $p_i>0, ~\forall 1\leq i\leq n$, we have
            \begin{equation}
                H(\boldsymbol{p})\leq H(\boldsymbol{u}) \quad \text{ and } \quad
                H_\alpha(\boldsymbol{p}) \leq H_\alpha(\boldsymbol{u}).
            \end{equation}
            That is, as indicators of the uncertainty of a probability distribution, both Shannon entropy and Renyi entropy achieve their maximums with the uniform distribution.
  \item For two independent probability distributions  $\boldsymbol{p}=(p_1,p_2,\cdots,p_n)$ and $\boldsymbol{q}=(q_1,q_2,\cdots,q_n)$, one has
                \begin{equation}
                    H(\boldsymbol{p},\boldsymbol{q})=-\sum_{i,j}p_i q_j \log(p_i q_j)=H(\boldsymbol{p})+H(\boldsymbol{q})
                \end{equation}
                and
                \begin{equation}
                    H_\alpha(\boldsymbol{p},\boldsymbol{q})=H_\alpha(\boldsymbol{p})+H_\alpha(\boldsymbol{q}).
                \end{equation}
\end{enumerate}

\begin{remark} In a given probability distribution $\boldsymbol{p}$, the smaller a component $p_i$ is, the less it contributes to the uncertainty of the distribution, i.e., it is clear that the corresponding event is unlikely to occur.
   Its contribution to the Shannon entropy is also small since $\lim_{p_i\rightarrow0} -p_i\log p_i=0$.
For those larger elements, although their probabilities of occurrence is larger, their contribution to the uncertainty and the entropy is also small, since $\lim_{p_i\rightarrow1} -p_i\log p_i=0$.
   In short,  compared with those events with probabilities  near from $\frac1n$, is much easier to predict the occurrence of those events with very large/small probabilities.
This is also true for  Renyi entropy.
\end{remark}

However, there are also situations where those small probability events are more concerned.
   To this end, we introduce a new measure emphasizing the importance of small probability events.

\subsection{The Message Importance Measure}
In this subsection, we shall introduce a new parametric information measure, which is referred to as the \textit{Message Importance Measure} (MIM).
\begin{definition}
        For a given probability distribution $\boldsymbol{p}=(p_1,p_2,\cdots,p_n)$ of finite alphabet, the message importance measure with parameter $\varpi$ is defined as
\begin{equation}
        L(\boldsymbol{p},\varpi)= L_\varpi(\boldsymbol{p})=\log\sum_{i=1}^{n} p_i \exp\{\varpi(1-p_i)\}
\end{equation}
where $\varpi\geq 0$ is the importance coefficient.
\end{definition}

\begin{remark}
Note that the larger $\varpi$ is, the larger contribution a small probability event has to the MIM.
    Thus, to manifest the importance of those small probability events, $\varpi$ is often chosen to be quite large, e.g., $\varpi=10$.
\end{remark}

\subsection{Outline of the Paper}
The rest of this paper is organized as follows.
    In Section \ref{sec:2_property}, we discuss the properties of the parametric message importance measure, including its extreme limit, convexity and the its relationship to the event decomposition/merging. The the selection of importance coefficient $\varpi$ is also discussed.
In Section \ref{sec:3_application}, we apply the message importance measure to the minority subset detection problem.
Section \ref{sec:4_discus} discuss the connection between the message importance measure and the Binary hypothesis testing problem.
    We present some simulation results to certificate the application of message importance measure in Section \ref{sec:5_sim} and finally, we conclude the paper in Section \ref{sec:6_conclusion}.

\section{The Properties of MIM} \label{sec:2_property}
In this section, the basic properties of the message importance measure is investigated in details.

\subsubsection{The Non-negative Property} The MIM $L(\boldsymbol{p},\varpi)$ is non-negative for any probability $\boldsymbol{p}$ and importance coefficient $\varpi\geq0$.

Note that for each element $p_i>0$ of the distribution, we have $p_i \exp\{\varpi(1-p_i)\} \geq p_i$
so that
\begin{equation}
    L(\boldsymbol{p},\varpi)=\log\sum_{i=1}^{n} p_i \exp\{\varpi(1-p_i)\}\geq \log\sum_{i=1}^{n} p_i =0.
\end{equation}

\subsubsection{The MIM of Uniform Distribution}  For the uniform distribution $\boldsymbol{u}=(1/n, 1/n, \cdots, 1/n)$, we have
\begin{equation}
    L(\boldsymbol{u},\varpi)=\varpi\left(1-\frac{1}{n}\right).
\end{equation}

\subsubsection{MIM Lower Bound}
For any probability distribution $\boldsymbol{p}=(p_1,p_2,\cdots,p_n)$ without zero elements, we have
\begin{equation} \label{rt:lowerbound}
    L(\boldsymbol{p},\varpi)\geq \varpi\left(1-\sum_{i=1}^n p_i^2\right)
\end{equation}

\begin{proof}
Define $f(x)=\exp\{-\varpi x\}$. It is readily seen that $f(x)$ is a convex function of $x\in \textbf{R}$.  According to Jensen's inequality, we have
\begin{equation}
    \mathbb{E}(f(X))\geq f(\mathbb{E}(X)),
\end{equation}
where $\mathbb{E}(\cdot)$ is the expected operation and $X$ is an arbitrary random variable.
    Assume that $X$ is drawn from the set $\{p_1,p_2,\cdots,p_n\}$ and follows the distribution $\boldsymbol{p}=\{p_1,p_2,\cdots,p_n\}$, we have
\begin{equation}
    \sum_{i=1}^{n} p_i \exp\{-\varpi p_i\} \geq  \exp\{-\sum_{i=1}^n \varpi p_i^2  \}.
\end{equation}

With some mathematical manipulations, one gets
\begin{equation}
    \log \sum_{i=1}^{n} p_i\exp\{\varpi(1-p_i)\}\geq \varpi\left(1-\sum_{i=1}^n p_i^2\right).
\end{equation}

In particular, the equality holds if and only if all the $p_i$ are equal, i. e. $p_i=1/n$.
\end{proof}

\subsubsection{The Maximum Value Property}
For any distribution  $\boldsymbol{p}=(p_1,p_2,\cdots,p_n)$ without zero elements,  if $\varpi \max_{i}{p_i}<2$ is satisfied, then we have
\begin{equation}
L(\boldsymbol{p},\varpi)\leq L(\boldsymbol{u},\varpi)
\end{equation}

\begin{proof}
Define the Lagrange as $g(\boldsymbol{p},\lambda) = \sum_{i=1}^n p_i\exp\{\varpi(1-p_i)\}+\lambda (\sum_{i=1}^n p_i - 1)$ for $\varpi>0$ and $x\in \textbf{R}$.
    It is readily seen that the partial derivative of $g(\boldsymbol{p},\lambda)$ with respect to $p_i$ is
\begin{equation}
    \frac{\partial g}{\partial p_i}=\exp\{\varpi(1-p_i)\}(1-\varpi p_i)+\lambda.
\end{equation}

By setting $\frac{\partial g}{\partial p_i}=0$ and recalling that $\sum_{i=1}^n p_i =1$, it can be readily testified that
$p_1=p_2=\cdots=p_n=1/n$ is the solution to the equations, which implies that the extreme value of $g$ can be achieved by the uniform distribution.

In addition, the second order derivative of $g(\boldsymbol{p},\lambda)$ with respect to $p_i$  is
\begin{equation}
\frac{\partial^2 g}{\partial p_i^2}= -\varpi\exp\{\varpi(1-p_i)\}(2-\varpi p_i)
\end{equation}

Therefore, if  $\varpi\max_i\{p_i\}<2$ is true, we have
\begin{equation}
    \frac{\partial^2 g}{\partial p_i^2}<0,
\end{equation}
which means that the uniform distribution reaches the maximum of $L(\boldsymbol{p},\varpi)$, i.e.,
$L(\boldsymbol{p},\varpi)\leq L(\boldsymbol{u},\varpi)$ for any $\boldsymbol{p}\neq\boldsymbol{u}$.
\end{proof}

\begin{remark}
    The The maximum value property of MIM is similar to the property (3) of Shannon Entropy and Renyi Entropy.
        However, they are actually different in that the maximum value property of MIM is conditioned on $\varpi\max_i{p_i}<2$, while that of Shannon and Renyi entropy is unconditional.
\end{remark}

\subsubsection{The Convexity Property}
For two given probability distributions $\boldsymbol{p}$ and $\boldsymbol{q}$ without zero elements, if $\varpi\max_i\{p_i,q_i\}<2$, then we have
\begin{equation}
    L(\alpha \boldsymbol{p} +(1-\alpha)\boldsymbol{q}, \varpi)\geq \alpha L(\boldsymbol{p},\varpi)+(1-\alpha) L(\boldsymbol{q},\varpi)
\end{equation}
for any  $0\leq \alpha \leq 1$.

\begin{proof}
Define $f(x)=x\exp\{\varpi(1-x)\}$ for some $\varpi>0$ and $x\in \textbf{R}$. The first order and the second order derivative of $f(x)$ are given by, respectively
\begin{eqnarray}
    f'(x) \hspace{-3mm}&=&\hspace{-3mm}\exp\{\varpi(1-x)\}(1-\varpi x), \\
    f''(x)\hspace{-3mm}&=&\hspace{-3mm}-\varpi\exp\{\varpi(1-x)\}(2-\varpi x).
\end{eqnarray}

It is clear that $f''(x)<0$ and $f(x)$ is concave in $x$ if $\varpi x\leq 2$.
    By using Jensen's inequality for the case of $\varpi\max_i\{p_i,q_i\}<2$, we have
\begin{eqnarray}
     \hspace{-3mm}& &\hspace{-3mm} \sum_i\big(\alpha p_i+(1-\alpha)q_i\big)\exp\big\{\varpi\big(1-\alpha p_i-(1-\alpha)q_i\big)\big\} \\
     \hspace{-3mm}&&\hspace{-3mm} \geq \sum_i \alpha p_i\exp\{\varpi(1-p_i)\}+(1-\alpha)q_i\exp\{\varpi(1-q_i)\}\\ \nonumber
      \hspace{-3mm}&&\hspace{-3mm} = \hspace{-0.5mm} \alpha \hspace{-0.5mm} \sum_ip_i\exp\{\varpi(1-p_i)\} \hspace{-0.5mm} + \hspace{-0.5mm} (1-\alpha) \sum_i q_i\exp\{\varpi(1-q_i)\}.
\end{eqnarray}

By the concavity of $\log(x)$ function for $x>0$ and Jensen' inequality, we have
\begin{align}
L(\alpha \boldsymbol{p} +(1-\alpha)\boldsymbol{q}, \varpi)\geq \alpha L(\boldsymbol{p},\varpi)+(1-\alpha) L(\boldsymbol{q},\varpi)
\end{align}
for any $0\leq \alpha \leq 1$, which proves the property.
\end{proof}

\subsubsection{Independent Probability Distributions}
For two independent probability distributions  $\boldsymbol{p}=(p_1,p_2,\cdots,p_n)$ and $\boldsymbol{q}=(q_1,q_2,\cdots,q_n)$, we have
\begin{equation}
L(\boldsymbol{pq}, \varpi)\leq L(\boldsymbol{p},\varpi)+H(\boldsymbol{p},\varpi).
\end{equation}

\begin{proof}
By the definition of MIM, we have
\begin{equation}
L(\boldsymbol{pq},\varpi)=\log \sum_{i,j} p_iq_j \exp\{\varpi(1-p_iq_j)\}~ \text{and}
\end{equation}
\begin{equation}
L(\boldsymbol{p},\varpi)+L(\boldsymbol{q},\varpi)=\log\sum_{i,j}p_iq_j \exp\{\varpi(2-p_i-q_i)\}.
\end{equation}

It can be readily verified that
\begin{align}
2-p_i-q_i-(1-p_iq_i)=(1-p_i)(1-q_i)>0,
\end{align}
which implies $L(\boldsymbol{pq}, \varpi)\leq L(\boldsymbol{p},\varpi)+H(\boldsymbol{q},\varpi)$.
\end{proof}

\begin{remark}
It is interesting to mention that the sum of the Shannon or Renyi entropies of two independent distributions equals to the corresponding entropy of the sum distribution.
    However, for the MIM, the equality holds only for the trivial case $p=q=1$.
By the definition of MIM, when one has collected all the information from different ways,  the total importance quantity is less than  the sum of the measured importance quantity of the different parts.
    On one hand, the information collector estimates the information importance more accurately by using the expected sum of each individual information observer, which reduces the information importance quantity that the collector can obtain.
On the other hand, this property indicates that the MIM is much more from the information coding and transmission. That is, traditional encoding techniques or ideas can not be used or be suitable to the information importance measure.
\end{remark}

\subsubsection{Event Decomposition and Merging}
For a given distribution  $\boldsymbol{p}=(p_1,p_2,\cdots,p_n)$ without zero elements, we have the following conclusions:
 \begin{enumerate}
   \item[a)] if the $i$-th event is divided into two sub-events $i^{(1)}$-th and $i^{(2)}$-th, the corresponding MIM will be increased;
   \item[b)] if the $i$-th event and $j$-th event are merged into a single event, the the corresponding MIM will be decreased.
 \end{enumerate}

 \begin{proof}
  Denote $p_i^{(1)}$ and $p_i^{(2)}$ as the probabilities of the first and the second sub-events of the $i$-th event, i.e., $p_i^{(1)}+p_i^{(2)}=p_i$, we have
 \begin{eqnarray}
    \hspace{-3mm}&&\hspace{-3mm} p_i^{(1)}\exp\{\varpi(1-p_i^{(1)})\}+ p_i^{(2)}\exp\{ \varpi(1 - p_i^{(2)})\}\\
        \hspace{-3mm}&&\hspace{-3mm} \geq p_i\exp\{\varpi(1-p_i)\}.
 \end{eqnarray}

 By take the sum over all the other events of the MIM, the first part of the property is proved.
    Moreover, the second part of the property is a direct result by reversing the event decomposition.
 \end{proof}

\begin{remark}
 This property indicates that the more observations we have, the more knowledge one can extract from the events involved in the messages.
 \end{remark}

\section{Application to Minority Subsets Detection}\label{sec:3_application}
In this section, we apply the proposed message importance measure to the minority subset detection problem.
    Note that MIM is actually the logarithm of the mean value of function $f(x)=x\exp\{\varpi(1-x)\}$ for $0<x<1$ and some $\varpi>0$.
Before we proceed, a useful Lemma on $f(x)$ is introduced.

 \begin{lemma}
    $f(x)$ achieves its maximum at $x=\frac1\varpi$. In particular, $f(x)$ is monotonically increasing for $0<x<\frac1\varpi$ and  is monotonically decreasing for $\frac1\varpi\leq x <1$.
 \end{lemma}

 \begin{proof}
 By taking the derivative of $f(x)$ with respective to $x$, we have
 \begin{eqnarray}
    f'(x) =  (1-\varpi x)\exp\{\varpi(1-x)\},
 \end{eqnarray}
 which proves the lemma immediately.
 \end{proof}

 \begin{lemma} \label{lem:2}
For a Bernoulli distribution $\boldsymbol{p}=(p, 1-p)$ where $0< p_0< p <\frac12$, there exists a $\varpi_0 >0$ such that for $\varpi \geq \varpi_0$, the MIM $L(\boldsymbol{p},\varpi)=\log \big(p\exp\{\varpi(1-p)\}+(1-p)\exp\{\varpi p\}\big)$ with importance coefficient $\varpi$ is strictly decreasing with $p$.
    Thus, the binary uniform distribution has the smallest message importance quantity when $\varpi>\varpi_0$,  i.e., $L(\boldsymbol{p},\varpi) > L(\boldsymbol{u}, \varpi)$ where $\boldsymbol{u}=(\frac12 ,\frac12 )$.
 \end{lemma}

 \begin{proof}
 Consider the difference
 \begin{eqnarray}\nonumber
\hspace{-3mm}&&\hspace{-3mm} L(\boldsymbol{p}, \varpi)-L(\boldsymbol{u}, \varpi)\\ \nonumber
\hspace{-3mm}&&\hspace{-3mm}= \log \big( p\exp\{\varpi(1-p)\} \hspace{-0.7mm}+ \hspace{-0.7mm}(1-p)\exp\{\varpi p\} \big) \hspace{-0.7mm} - \hspace{-0.7mm}\log \exp\left\{\varpi/2\right\}\\ \nonumber
 \hspace{-3mm}&&\hspace{-3mm}=\log \left( p\exp\left\{\varpi\left(1/2 -p\right)\right\}+(1-p)\exp\{\varpi(p-1/2 )\}\right)\\
 \hspace{-3mm}&&\hspace{-3mm}\triangleq \log\psi(p) .
 \end{eqnarray}

 The derivative of $\psi (p)$ with respective to $p$ is given by
 \begin{eqnarray}
    \psi'(p) \hspace{-3mm}&=&\hspace{-3mm} \exp\left\{\varpi\left(1/2 -p\right)\right\} (1-\varpi p) \\
                 \hspace{-3mm}&&\hspace{-3mm}  +\exp\left\{\varpi\left(p-1/2 \right)\right\}(\varpi-1-\varpi p).
 \end{eqnarray}

 Since $p_0< p < \frac12 $, the exponential parts in the first term is positive and the second term is  negative, respectively.  That is,
 \begin{eqnarray}
 \hspace{-3mm}&&\hspace{-3mm}\lim_{\varpi \rightarrow \infty} \exp\{\varpi(1/2 -p)\}(1-\varpi p)= -\infty, \\
\hspace{-3mm}&&\hspace{-3mm} \lim_{\varpi\rightarrow \infty} \exp\{\varpi(p-1/2 )\}(\varpi-1-\varpi p)=0.
 \end{eqnarray}

 It is clear that there exists some $\varpi_0(p_0)$ such that for $\varpi \geq \varpi_0(p_0)$, $\psi(p_0)$ is strictly decreasing with respect to $p$.
    Thus we have $\psi(p)>\psi(\frac12 )=1$, which leads to
 \begin{equation}
 L(\boldsymbol{p}, \varpi)-L(\boldsymbol{u}, \varpi)\geq 0.
 \end{equation}

This proves the Lemma.
  \end{proof}

\begin{remark}
The Lemma indicates that for the Bernoulli distribution, if importance coefficient $\varpi$ is properly selected, the uniform distribution has the least importance quantity.
    This is very different from the conventional information measures such as Shannon entropy and Renyi entropy.
In fact, the uniform distribution has the largest uncertainty, but may have lest importance since its distribution mode are too popular.
\end{remark}

\begin{remark}
In the conventional source encoding of information theory, those typical sets (events with relatively large occurring probabilities) are more important than those atypical sets (events with relatively small occurring probabilities).
     In big data, especially for those minority subsets detections, however, the typical sets are less important than those atypical sets.
Actually, for the minority subset detections in big data era, the proposed parametric information measure MIM, which can reflect the social values of the atypical sets, may have much potential.
\end{remark}

  \begin{lemma}
   For a probability distribution $\boldsymbol{p}=(p_1,p_2,\cdots,p_n)$ with $0<\min_i\{p_i\}$, there exists a $\varpi_0 >0$ such that for $\varpi \geq \varpi_0$,  the parametric message importance measure satisfies
   \begin{align}
 L(\boldsymbol{p}, \varpi)=\log \sum_{i=1}^n p_i\exp\{\varpi(1-p_i)\} > L(\boldsymbol{u}, \varpi),
 \end{align}
 where $\boldsymbol{u}=(1/n, 1/n,\cdots,1/n)$ is the uniform distribution.
    When the importance coefficient $\varpi$ is sufficient large, we further have
 \begin{equation} \label{rt:m_lbound}
 L(\boldsymbol{p}, \varpi)\doteq \varpi(1-p_{\min}) +\log p_{\min}.,
 \end{equation}
 where $p_{\min}=\min_i\{p_i\}$. 
  \end{lemma}

\begin{proof}
By considering the difference
\begin{eqnarray}
\hspace{-3mm}&&\hspace{-3mm} L(\boldsymbol{p},\varpi)-L(\boldsymbol{u},\varpi) \\
\hspace{-3mm}&&\hspace{-3mm} =\log \sum_i p_i \exp\{\varpi(1/n-p_i)\} \\ \label{eq:43}
\hspace{-3mm}&&\hspace{-3mm} =\log \Big(\sum_{p_i\geq 1/n}p_i \exp\{\varpi(1/n-p_i)\} \\
\hspace{3mm}&&\hspace{11mm}      +\sum_{p_i < 1/n}p_i \exp\{\varpi(1/n-p_i)\} \Big).
\end{eqnarray}

It clear in (\ref{eq:43}) that the first sum item decreases exponentially with the increase of $\varpi$ while the second sum item is increasing with $\varpi$.
    Therefore, as $\varpi$ is increased to sufficiently large, the second sum item dominates the value of the difference. In particular,

\begin{eqnarray}
\hspace{-10mm}&&\hspace{-3mm} \lim_{\varpi \rightarrow \infty}
                                \frac{\log \sum_i p_i \exp \{\varpi(1/n-p_i)\}} {\log\big(p_{\min}\exp\{\varpi(1/n-p_{\min})\}\big)} \\
\hspace{-10mm} &&\hspace{-3mm} = \hspace{-0.8mm} \lim_{\varpi \rightarrow \infty}
                                \frac{ \varpi(1/n \hspace{-0.6mm} - \hspace{-0.6mm} p_{\min}) \hspace{-0.6mm} + \hspace{-0.6mm} \log \sum_{i=1}^{n}p_i \exp\{ \varpi(p_{\min} \hspace{-0.6mm} - \hspace{-0.6mm} p_i) \} }
                                        { \varpi(1/n-p_{\min})+ \log p_{\min}  } \\
\hspace{-10mm} &&\hspace{-3mm} = \hspace{0.3mm}1.
\end{eqnarray}

Therefore, when the importance coefficient $\varpi$ is sufficient large,  we have
 \begin{eqnarray}
    L(\boldsymbol{p}, \varpi)\hspace{-3mm} &\doteq&\hspace{-3mm} L(\boldsymbol{u},\varpi) \hspace{-0.5mm} + \hspace{-0.5mm}
                                                                                                                             \log\big(p_{\min}\exp\{\varpi(1/n \hspace{-0.5mm}-\hspace{-0.5mm}p_{\min})\}\big)\\
                                                \hspace{-3mm} &\doteq&\hspace{-3mm}  \varpi(1-1/n) +\log p_{\min} + \varpi(1/n-p_{\min})  \\
                                                \hspace{-3mm} &\doteq&\hspace{-3mm}  \varpi(1-p_{\min}) +\log p_{\min}.
 \end{eqnarray}
which proves the lemma.

\end{proof}

\begin{remark}
This lemma indicates that when the importance coefficient $\varpi$ of MIM is selected proper large, it can be used to dig out the meaning of those events with small probabilities.
    This also validates why MIM can be used to the minority subset detections in big data.
\end{remark}

\subsection{Binary Minority Subset Detection} \label{sec:3a}

  Consider the scenario with two possible events, in which one occurs with a very small probability $0<p\ll 1$ and the other occurs with a much larger probability $1-p$, i.e., $\boldsymbol{p}=\{p,1-p\}$.
    For this case, the corresponding MIM is
  \begin{eqnarray}
        L(\boldsymbol{p}, \varpi) \hspace{-3mm}&=&\hspace{-3mm}  \log \big(p\exp\{\varpi(1-p)\} \hspace{-0.5mm}+\hspace{-0.5mm} (1-p)\exp\{\varpi p\}\big) \\
                                                       \hspace{-3mm}&=&\hspace{-3mm} \log\big(1-p+p\exp\{\varpi(1-2p)\}\big)+\varpi p\\ \label{rt:L_apx}
                                                      \hspace{-3mm}&\doteq&\hspace{-3mm}   \log\big(1+p\exp\{\varpi(1-2p)\}\big)+\varpi p,
  \end{eqnarray}
which can be used to evaluate the importance for those events with very small occurring probabilities.

 As known to all, hypothesis testing is a very important technique to discriminate events based on the logarithmic maximum likelihood ratios.
    In general,  the prior probability of each event are assumed to be known or can be predicated exactly.
 However, there are also many scenarios in which the  probability of some events within the minority part can not be known exactly, e.g., only a rough range of the probability is given.
    In this case, further estimation of the prior probabilities is needed in order to employ the hypothesis testing technique.
 To this end, we present a method to determine these prior probabilities based on the proposed  message importance measure.

For a binary hypothesis testing problem, we denote $H_0$ and $H_1$ as the hypotheses for ``event 0" and ``event 1", where "event 0" is the minority part with a very small occurring probability.
    It is only known that the occurring  probability of ``event 0" satisfies $p_0^{(1)}\leq P(H_0)\leq p_0^{(2)}$.
Our problem is to estimate the prior probability of $H_0$ for the hypothesis testing, under a certain optimal criteria.

In this paper,  this problem will be solved through an induction process, which consists of two steps.
    In the first step, the importance coefficient $\varpi$ of the MIM is estimated under a proper estimation criterion.
In the second step, the prior probability $p_0$ is obtained based on $\varpi$ and the corresponding message importance measure.
    The details are given as follows.

\subsubsection{Step 1}
 To be fair, one can think that the given probability bounds $p_0^{(1)}$ and $p_0^{(2)}$ are assumed to be equal important for the evaluation of $P(H_0)$, so we have
 \begin{equation} \label{eq:step1}
     L([p_0^{(1)},1-p_0^{(1)}], \varpi)=L([p_0^{(2)},1-p_0^{(2)}], \varpi).
 \end{equation}

 Since the  $p_0^{(1)}$ and $p_0^{(2)}$ are both small numbers near zero,  by employing (\ref{rt:L_apx}), the the two sides of  (\ref{eq:step1}) can be approximated by,
  \begin{eqnarray}
  \hspace{-6mm} &&\log(1+p_0{(1)}\exp\{\varpi(1-2p_0^{(1)}\})+\varpi p_0^{(1)}\\
  \hspace{-6mm} &&  =\log(1+p_0{(2)}\exp\{\varpi(1-2p_0^{(2)}\})+\varpi p_0^{(2)}.
  \end{eqnarray}

  It yields
  \begin{eqnarray}
  \hspace{-6mm}  (p_0^{(2)}-p_0^{(1)})\varpi \hspace{-3mm} &=& \hspace{-3mm}
                                                                \log \frac{1+p_0{(1)}\exp\{\varpi(1-2p_0^{(1)})\}}{1+p_0{(2)}\exp\{\varpi(1-2p_0^{(2)})\}} \\
                                                    \label{eq:56}
                                                    \hspace{-3mm} &=& \hspace{-3mm}  \log\frac{p_0^{(2)}}{p_0^{(1)}},
  \end{eqnarray}
where (\ref{eq:56}) is because $\exp\{\varpi(1-2p_0^{(1)})\}\gg 1$ and $\exp\{\varpi(1-2p_0^{(2)})\}\gg 1$.

Therefore, we have
   \begin{equation}
   \varpi=\frac{\log p_0^{(2)}-\log p_0^{(1)}}{p_0^{(2)}-p_0^{(1)}}.
   \end{equation}

  \subsubsection{Step 2}
By Lemma 1, we know that the optimal estimation of the prior probability of $H_0$ under the MIM criteria is
   \begin{eqnarray}
        \hat{p}_0 \hspace{-3mm}&=&\hspace{-3mm} 1/\varpi\\ \label{rt:p_0}
            \hspace{-3mm}&=&\hspace{-3mm} \frac{p_0^{(2)}-p_0^{(1)}} {\log p_0^{(2)}-\log p_0^{(1)}}.
   \end{eqnarray}

Specifically, $\hat{p}_0 $  is obtained by maximizing the MIM while balancing the fairness between the two estimation boundaries  $p_0^{(1)} $ and $p_0^{(2)}$.
   In particular, the following lemma further validates the reasonableness of the two-step predication of the prior probability $P(H_0)$.

   \begin{lemma} \label{lem:abc}
        For any $0<a<b<1$,  there exists a real number $c$ given by
   \begin{equation} \label{rt:abc}
        c=\frac{b-a}{\log b -\log a}
   \end{equation}
   satisfying  $a<c<b$.
   \end{lemma}

   \begin{proof}
   The lemma is proved by contradiction.
        For a $c$ given by (\ref{rt:abc}), we first assume that $c>b$.

   Define $g(x)=\log x -\log a -\frac{x-a}{c}$, it is readily seen that $g(a)=g(b)=0$.
Note that the derivative with respect to $x$ is given by
   \begin{equation}
        g'(x)=\frac{1}{x}-\frac{1}{c}.
   \end{equation}

   It is clear that  $g(x)$ is strictly monotonically increasing with $x$ for any $0<x<c$.
        By the assumption that $c>b$, that if $x<c$, we have $g(a)<g(b)$, which is contradict with the fact that $g(a)=g(b)=0$.
Thus, we must have $c<b$.
   Likewise, one can prove that $c>a$ holds.
By combing the two cases, the lemma is proved.
   \end{proof}

   Based on Lemma \ref{lem:abc}, it is clear  that the estimated probability $\hat{p}_0$ (see (\ref{rt:p_0})) satisfies $ p_0^{(1)}<\hat{p}_0<p_0^{(2)}$.
    This is consistent with one's expectation.
    In fact, such a estimation also makes sense in the view point of accuracy, owing to  $\lim_{p_0^{(1)}\rightarrow p_0^{(2)}} \frac{p_0^{(2)}-p_0^{(1)}} {\log p_0^{(2)}-\log p_0^{(1)}} = p_0^{(2)}$.
    That is, if the prior estimated values of the probability on $H_0$ is accurate (bounds $p_0^{(1)}$ and $p_0^{(2)}$ are tight), the estimated probability value  given by our two-step method will obtain the same result. This certificates the usefulness of our defined message importance measure in theory.

   \subsection{$M$-ary Minority Subset Detection} \label{sec:3b}
    $M$-ary hypothesis testing is a generalization of binary hypothesis, which has been widely used in signal detections, pattern recognition and group detections.
        The minority subset detection for $M$-ary hypothesis detecting is described as follows.
    For a given probability distribution $\boldsymbol{p}=\{p_1,p_2,\cdots, p_n\}$, we want to select a proper importance coefficient $\varpi$ for the MIM measure $L(\boldsymbol{p}, \varpi)$ so that  it can be used to characterize the minority subsets importance while reducing the effects of  uniform distribution.
        To be specific, we consider the problem of finding an  importance coefficient $\varpi$ satisfying
    \begin{align}
            L(\boldsymbol{p}, \varpi)\geq L(\boldsymbol{u}, \varpi).
    \end{align}

 By using the events decomposition and merging property (property 7), we have
 \begin{eqnarray} \nonumber
    L(\boldsymbol{p}, \varpi)\geq \log\big(p_{\min}\exp\{ \hspace{-3mm}&&\hspace{-3mm}\varpi(1-p_{\min})\}\\
          \hspace{-3mm}&&\hspace{-3mm} + (1-p_{\min})\exp\{\varpi p_{\min}\}\big)
 \end{eqnarray}
 where $p_{\min}=\min_i\{p_i\}$.

 It is clear that the solution $\varpi$ satisfies
 \begin{eqnarray}
\hspace{-6mm} \log\hspace{-3mm}&&\hspace{-3mm} \big( p_{\min}   \exp\{\varpi(1-p_{\min})\}+(1-p_{\min})\exp\{\varpi p_{\min}\}\big) \\
                                      \hspace{-3mm}&&\hspace{-3mm} -\log\big(\exp\{\varpi(1-1/n)\}\big)\geq 0.
 \end{eqnarray}

 After some manipulations, we have
 \begin{eqnarray}\nonumber
\hspace{-6mm}    \log\big((1-p_{\min})+p_{\min}\exp\{\hspace{-3mm}&&\hspace{-3mm}\varpi(1-2p_{\min})\}\big) \\
                                                            \hspace{-3mm}&&\hspace{-3mm}\geq \varpi(1-p_{\min}-1/n).
 \end{eqnarray}

 By choosing such a $\varpi$ that
 \begin{align}
    \log\big(p_{\min}\exp\{\varpi(1-2p_{\min})\}\big) \geq \varpi(1-p_{\min}-1/n),
 \end{align}
we have
 \begin{align}
        \varpi \geq -\frac{\log p_{\min}}{1/n-p_{\min}}.
 \end{align}

 Based on the discussion, we have the following theorem.

 \begin{theorem} \label{th:m_ary}
    Given a probability distribution $\boldsymbol{p}=\{p_1,p_2,\cdots, p_n\}$ without zero elements, if $\varpi$ satisfies
 \begin{align}
        \varpi \geq -\frac{\log p_{\min}}{1/n-p_{\min}},
 \end{align}
 then we have
 \begin{equation}
  L(\boldsymbol{p}, \varpi)\geq L(\boldsymbol{u}, \varpi),
  \end{equation}
  where $p_{\min}=\min_i \{p_i\}$ and $\boldsymbol{u}=\{1/n,1/n,\cdots,1/n\}$ is the uniform distribution for $n>2$.
 \end{theorem}

 \begin{proof}
 The proof follows the aforementioned discussion.
 \end{proof}

Theorem \ref{th:m_ary} indicates that if the importance coefficient $\varpi$ is selected sufficiently large, the effect of uniform distribution can be reduced and the minority subset dominates the value of the MIM function, which also reflects the information importance of those minority subsets in the system detection, especially in big data.
    Moreover, it is seen that the reasonable range of importance coefficient $\varpi$ for the minority subset detection also depends on the minimum probability $p_{\min}$ of events. 

\begin{theorem} \label{th:mary_2}
     Given a probability distribution $\boldsymbol{p}=\{p_1,p_2,\cdots, p_n\}$ without zero elements, for each $p_{s}$ satisfying $p_s<1/n$, if $\varpi$   satisfies
 \begin{align}
        \varpi \geq -\frac{\log p_{s}}{1/n-p_{s}},
 \end{align}
 then we have
 \begin{equation}
  L(\boldsymbol{p}, \varpi)\geq L(\boldsymbol{u}, \varpi),
  \end{equation}
  where  $\boldsymbol{u}=\{1/n,1/n,\cdots,1/n\}$ is the uniform distribution for $n>2$.
 \end{theorem}

 \begin{proof}
According to the events decomposition and merging property (see property 7), for each $p_s$ satisfying $p_s<1/n$, we have
 \begin{align}
\hspace{-1mm}        L(\boldsymbol{p}, \varpi) \hspace{-0.6mm}\geq \log\big(p_{s}\exp\{\varpi(1 \hspace{-0.6mm} - \hspace{-0.6mm} p_{s}\} \hspace{-0.6mm} +\hspace{-0.6mm} (1\hspace{-0.6mm}-\hspace{-0.6mm}p_{s}) \hspace{-0.6mm} \exp\{\varpi p_{s}\}\big).
 \end{align}

By substituting $p_{\min}$ with $p_{s}$, it is seen that if $\varpi$ satisfies
 \begin{align}
        \varpi \geq -\frac{\log p_{s}}{1/n-p_{s}},
 \end{align}
 then we have
 \begin{equation}
  L(\boldsymbol{p}, \varpi)\geq L(\boldsymbol{u}, \varpi),
  \end{equation}
  This completes the proof.
  \end{proof}

 By summarizing the Theorem \ref{th:m_ary} and Theorem \ref{th:mary_2}, we have the following main theorem.
  \begin{theorem}\label{th:main}
    Given a probability distribution $\boldsymbol{p}=\{p_1,p_2,\cdots, p_n\}$ without zero elements, if $\varpi$  satisfies
 \begin{align}
    \varpi \geq \min_{p_i: \{p_i < 1/n\}} \left\{-\frac{\log p_{i}}{1/n-p_{i}}\right\},
 \end{align}
 then we have
 \begin{equation}
  L(\boldsymbol{p}, \varpi)\geq L(\boldsymbol{u}, \varpi),
  \end{equation}
  where  $\boldsymbol{u}=\{1/n,1/n,\cdots,1/n\}$ is the uniform distribution for $n>2$.
\end{theorem}

As indicated by Theorem \ref{th:main}, it is seen that $1/n$, is a critical threshold for the elements of a probability distribution, which is not achievable in the selection of  the importance coefficient $\varpi$.
 Naturally, this raises another interesting problem: if the minimum probability $p_{\min}$  approaches $1/n$, what will happen?

According to Theorem \ref{th:main}, it is seen that as $p_{\min}$ approaches $1/n$, the required importance coefficient $\varpi$ goes to infinity, which is not an expected answer.
    Recall that the $\varpi$ must satisfy
\begin{eqnarray}\nonumber
\hspace{-3mm}    \log\big((1-p_{\min})+p_{\min}\exp\{\varpi\hspace{-3mm}&&\hspace{-3mm}(1-2p_{\min})\}\big) \\
            \hspace{-3mm}&&\hspace{-3mm}\geq \varpi(1-p_{\min}-1/n).
 \end{eqnarray}

 For a given $\varpi$,  as $p_{\min}$ approaches $1/n$, we have
 \begin{align}
    \log\hspace{-0.3mm}\big(1 \hspace{-0.3mm} -\hspace{-0.3mm}1/n)\hspace{-0.3mm}+\hspace{-0.3mm}
    1/n\exp\{\varpi(1\hspace{-0.3mm}-\hspace{-0.3mm}2/n)\}\big) \hspace{-0.3mm}
     \geq \hspace{-0.3mm} \varpi (1\hspace{-0.3mm}-\hspace{-0.3mm}2/n),
 \end{align}
which  is contradict to the fact because $  \log[(1-1/n)+1/n\exp\{(1-2/n)\}] <\varpi (1-2/n)$ holds for all $n>2$.
    This means that when $p_{\min}$ approaches to $1/n$, there does not exist any reasonable solution of  $\varpi$ to characterize the minority subset importance.
In other words, when a distribution is close to the uniform distribution, uniform distribution always has the maximum importance value,  for any finite $\varpi$.
    In this case, the  message importance measure works in a same way as the well known Shannon and Renyi entropies.

\setcounter{TempEqCnt}{\value{equation}} 
\setcounter{equation}{93} 
 \begin{figure*}[!t]
 \hrulefill
 \vspace*{10pt}
 \begin{equation}
        E\leq \min_{\alpha: 0<\alpha<1} \left\{\omega_0^\alpha (1-\omega_0)^{1-\alpha} \int p_0^{\alpha}(x)\left\{\sum_{k=1}^{M-1}\frac{\omega_k}{1-\omega_0}p_k(x)\right\}^{1-\alpha} dx\right\}.
 \end{equation}
 \vspace*{10pt}
   \hrulefill
  \end{figure*}
  \setcounter{equation}{\value{TempEqCnt}} 

   \section{Discussions}\label{sec:4_discus}
 \subsection{Binary Hypothesis Testing and Bayes Decision Rule}
The conventional binary hypothesis testing focus on determining a decision rule  which classifies  a random observation $x$ into one of the two possible classes of outcomes while minimizing the average decision error.
    We denote $H_0$ and $H_1$ as two event classes (hypothesises) with probabilities $\omega_0=P(H_0)$ and $\omega_1=P(H_1)=1-\omega_0$, respectively.
Let $p_0(x)=P(x|H_0)$ and $p_1(x)=P(x|H_1)$ the corresponding posterior probabilities.
    In this case, the Bayes decision rules classifies $x$ as $H_0$ if $\omega_0 p_0(x) >\omega_1 p_1(x)$, and the corresponding decision error is given by
 \begin{eqnarray}
    E \hspace{-3mm}&=&\hspace{-3mm}  \int \min\{P(H_0|x),P(H_1|x)\}p(x) dx\\
        \hspace{-3mm}&=&\hspace{-3mm}  \int \min\{\omega_0 p_0(x), \omega_1 p_1(x)\}dx,
 \end{eqnarray}
 where $P(H_0|x)=\frac{\omega_0 p_0(x)}{p(x)}$,  $P(H_1|x)=\frac{\omega_1 p_1(x)}{p(x)}$, and $p(x)$ is the probability distribution of $x$.

 By using the inequality
 \begin{align}
    \min\{a, b\} \leq a^\alpha b^{1-\alpha}, ~~ \forall \alpha \in (0,1),\text{and $a>0, b>0 $},
 \end{align}
we have
 \begin{align}
    E\leq \omega_0^\alpha \omega_1^{1-\alpha} \int p_0^{\alpha}(x) p_1^{1-\alpha}(x) dx
 \end{align}

 By further optimizing the estimation error over $\alpha$, we have
 \begin{align}
        E\leq \min_{\alpha: 0<\alpha<1} \left\{\omega_0^\alpha \omega_1^{1-\alpha} \int p_0^{\alpha}(x) p_1^{1-\alpha}(x) dx\right\}.
 \end{align}


The following example further explains this algorithm in details.
    In particular, we consider the following two hypothesises:
 $H_0$: $p_0(x)=\frac{1}{\sqrt{2\pi \sigma^2}} \exp\{-\frac{(x-\mu_0)^2}{2\sigma^2}\}$ and
 $H_1$: $p_1(x)=\frac{1}{\sqrt{2\pi \sigma^2}} \exp\{-\frac{(x-\mu_1)^2}{2\sigma^2}\}$.

 Then the decision error is bounded by
 \begin{eqnarray}
   \hspace{-6mm} E \hspace{-3mm}&\leq&\hspace{-3mm} \min_{\alpha: 0<\alpha<1} \left\{\omega_0^{\alpha} \omega_1^{1-\alpha}
                \exp \left\{-\frac{\alpha (1-\alpha)(\mu_0-\mu_1)^2}{2\sigma^2} \right\} \right\}\\ \nonumber
  \hspace{-3mm}&=&\hspace{-3mm} \min_{\alpha: 0<\alpha<1}
                \exp\Big\{ \alpha\log \omega_0 +(1-\alpha)\log \omega_1 \\
  \hspace{-3mm}&&\hspace{-3mm}  \qquad  \qquad  \qquad \qquad\qquad  -\frac{\alpha (1-\alpha)(\mu_0-\mu_1)^2}{2\sigma^2}\Big\}\\
  \label{rt:esti_e}
 \hspace{-3mm}&=&\hspace{-3mm} \min_{\alpha: 0<\alpha<1}  \hspace{-0.3mm}
                \exp\{ \alpha\log \omega_0 \hspace{-0.6mm} + \hspace{-0.6mm} (1 \hspace{-0.5mm} - \hspace{-0.5mm}\alpha)\log\omega_1
                  \hspace{-0.3mm} - \hspace{-0.3mm} \alpha (1 \hspace{-0.5mm} - \hspace{-0.5mm}\alpha)\beta\},
 \end{eqnarray}
 where $\beta=\frac{(\mu_0-\mu_1)^2}{2\sigma^2}$.\\


For notational simplicity, we denote the estimation error in (\ref{rt:esti_e}) as
\begin{equation}\label{eq:kalpha}
    K(\alpha)=\alpha\log \omega_0 +(1-\alpha)\log\omega_1 -\alpha (1-\alpha)\beta.
\end{equation}

 Calculating the  derivative of $K(\alpha)$ with respect to $\alpha$  and setting it to zero, we have
 \begin{equation}
     K'(\alpha)=\log \omega_0 -\log \omega_1 -\beta (1-2\alpha)=0,
 \end{equation}
for which the  solution $\alpha^{*}$ is given by
\begin{equation}
    \alpha^{*}=\frac{1}{2}+\frac{\log \omega_1 -\log \omega_0}{2\beta}.
\end{equation}

It is clear that when $\omega_0<\omega_1$, the optimal  $\alpha$ is larger than $\frac12 $.

Recall that for binary minority subset detection problem, if one only knows that $\omega_0\in(\omega_0^{(1)},\omega_0^{(2)})$ and $\omega_0^{(2)}\ll 1$, the above optimal processing algorithm can only proceed by
select the worst case, such as, $\omega_0=\omega_0^{(2)} $.
    In case that  $\omega_0^{(1)}$ is much smaller than $\omega_0^{(2)}$, the worst case processing can result in large bias to the decision error estimation.
Therefore, the Chernoff information method can not deal with such kind of minority subset detection problem.

 \subsection{M-ary Minority Subset Detection and Bayes Decision Rule}

 Following the discussion on binary minority subset detection problem, we consider the minority subset detection for the M-ary case in this subsection.

  Denote $H_0, H_1, \cdots, H_{M-1}$  as the $M$ classes with probabilities $\omega_k=P(H_k), (k=0,1,\cdots, M-1)$ and $p_k(x)=P(x|H_k)$.
        Without loss of generality, we assume that the minority subset is the class $H_0$.
  According to  the Bayes decision rule, $x$ is classified as $H_0$ if $\omega_0 p_0(x) > \sum_{k=1}^{M-1}\omega_k p_k(x)$.
        The corresponding decision error is given by
 \begin{eqnarray}
         E \hspace{-3mm}&\ = &\hspace{-3mm} \int \min\{P(H_0|x),P(\bar{H}_0|x)\}p(x) dx\\
             \hspace{-3mm}&\ = &\hspace{-3mm} \int \min\{\omega_0 p_0(x), \sum_{k=1}^{M-1} \omega_k p_k(x)\}dx,
 \end{eqnarray}
 where $P(H_0|x)=\frac{\omega_0 p_0(x)}{p(x)}$, $P(\bar{H}_0|x)=\frac{\sum_{k=1}^{M-1}\omega_k p_k(x)}{p(x)}$, and $p(x)$ is the probability distribution of $x$.

 Likewise, we have
 \begin{align}
        E\leq \omega_0^\alpha (1-\omega_0)^{1-\alpha} \int p_0^{\alpha}(x)\left\{\sum_{k=1}^{M-1}\frac{\omega_k}{1-\omega_0}p_k(x)\right\}^{1-\alpha} dx
 \end{align}
which can be further simplified into (94) as shown on the top of this page.
\setcounter{equation}{94}

 It is clear that if  one only has the knowledge that $\omega_0 \in (\omega_0^{(1)},\omega_0^{(2)})$ and $\omega_0^{(2)}\ll 1$, the conventional Bayes decision rule and Chernoff Information tools can not give a satisfying estimation of the minority subset probability $\omega_0$.
    On the contrary, the message importance measure proposed in this paper performs much better, as shown in Subsection \ref{sec:3a} and Subsection \ref{sec:3b}.

\begin{figure*}[htp]

\hspace{-6 mm}
    \begin{tabular}{cc}
    \subfigure[Message importance measure versus $\varpi$]{
    \begin{minipage}[t]{0.5\textwidth}
    \centering
    {\includegraphics[width = 84mm] {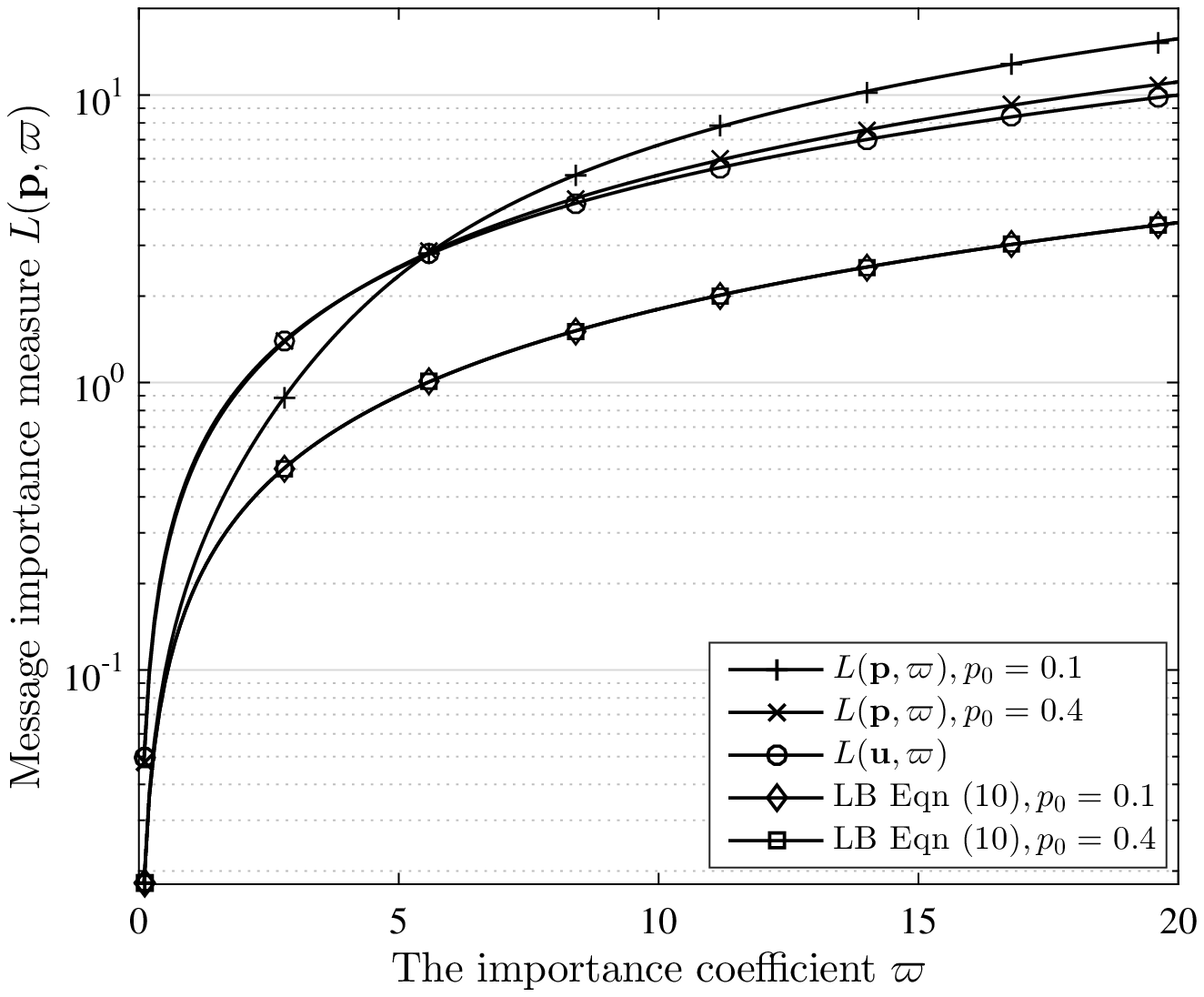} \label{fig:lpw}}
    \end{minipage}}

    \hspace{0.01\textwidth}
    \subfigure[Message importance measure versus $p_0$]{
    \begin{minipage}[t]{0.5\textwidth}
    \centering
    {\includegraphics[width = 84mm] {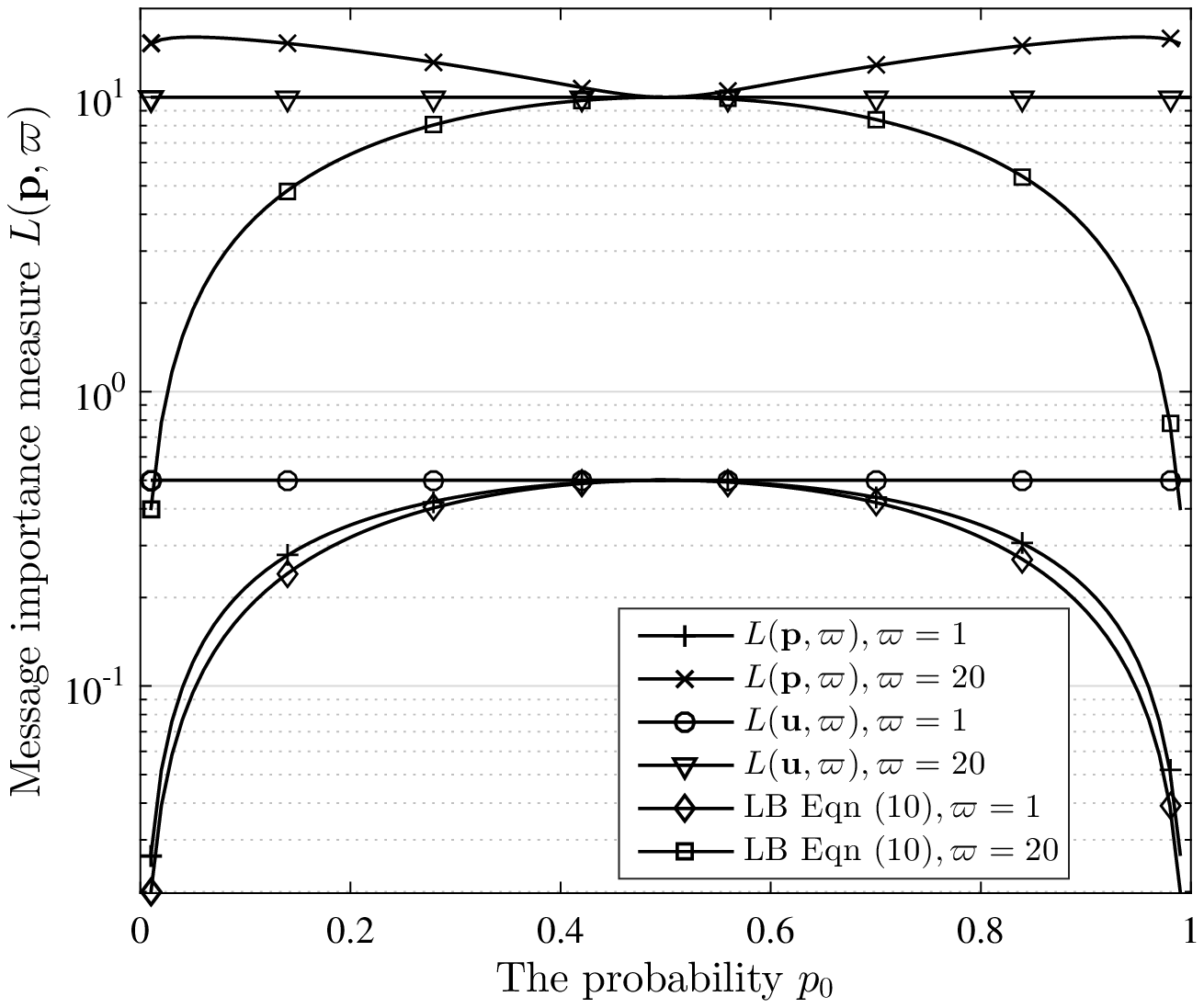} \label{fig:lpp}}
    \end{minipage}}\\

    \end{tabular}

\caption{The message importance measure in the binary case.} \label{fig:lwp}
\end{figure*}

 \section{Numerical Results} \label{sec:5_sim}
 In this section, we further illustrates the properties of message importance measure and its effectiveness through numerical results.

 \subsection{The property of MIM}
 Taking the Bernoulli distribution $\boldsymbol{p}_b=\{p_0,p_1\}$ and binary uniform distribution $\boldsymbol{u}_b=\{0.5, 0.5\}$ as  examples, we present  how MIM varies with importance coefficient $\varpi$ and probability $p_0$, as shown in Fig. \ref{fig:lwp}.

 From Fig. \ref{fig:lpw}, it is seen that $L(\boldsymbol{p}_b, \varpi)|_{\varpi=0.1}$ is larger than $L(\boldsymbol{u}_b, \varpi)$ if $\varpi>5.\dot{5}$, which is in agreement with property 4 and Lemma \ref{lem:2}.
    It is also noted that the lower bound given by (\ref{rt:lowerbound}) in property 3 is tighter when $\varpi$ is relatively small.
 Fig. \ref{fig:lpp} presents how the MIM changes with probability $p_0$.
    First, it is clear that the curve is symmetric about $p_0=0.5$.
 Second, when $\varpi$ is relatively small, i.e. $\varpi=1$, we have  $L(\boldsymbol{p}_b, \varpi)<L(\boldsymbol{u}_b, \varpi)$. On the contrary, $L(\boldsymbol{p}_b, \varpi)>L(\boldsymbol{u}_b, \varpi)$ for $\varpi=20$, which is relatively large.

 \subsection{Importance Coefficient $\varpi$ selection in M-ary  minority subset detection}

In this subsection, we investigate the selection of importance coefficient $\varpi$ according to Theorem 1--3.
    We consider a distribution $\boldsymbol{p}=[0.0925, 0.3156, 0.3887, 0.1484, 0.0549]$.
By Theorem 1, we have $\varpi_0=20.0011$ and $L(\boldsymbol{p},\varpi)>L(\boldsymbol{u},\varpi)$ when $\varpi>\varpi_0$, as shown in Fig. \ref{fig:3w0}.
    The performance of lower bound (\ref{rt:m_lbound}) is also presented in Fig. \ref{fig:3w0}.
        As seen, the lower bound is quite tight when $\varpi$ is large.

\begin{figure}[!t]
\centering
\includegraphics[width=3.5in]{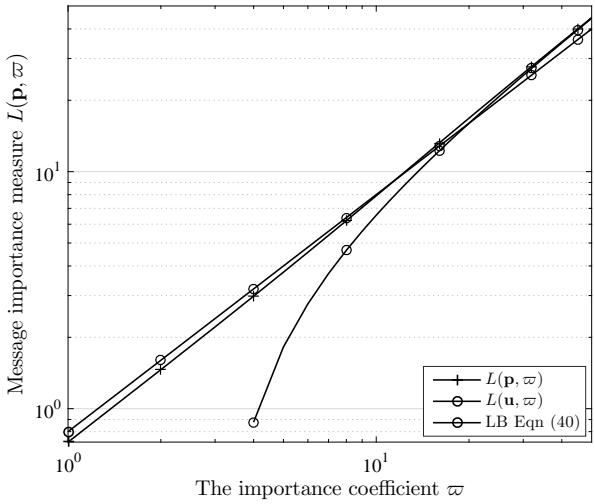}
\caption{The the selection of $\varpi$.} \label{fig:3w0}
\end{figure}

\section{Conclusion} \label{sec:6_conclusion}

In this paper, we investigated the message importance evaluation problem and proposed an new parametric information measure, i.e., the message importance measure, which keeps similar properties as Shannon entropy and Renyi entropy, in characterizing the event uncertainty.
    Moreover, the message importance measure has its own distinct properties by highlighting the importance of atypical events with small occurring probabilities.
This makes message importance measure a promising measure for the statistical processing of big data.
    We have investigated the major properties of message importance measure, and presented an binary detection algorithm.
We also have presented the selection rule for the importance coefficient $\varpi$ for more general minority subset detection problems.
    Designing better minority subset detection algorithms and  investigating their performance are of our future interests.

\section*{Acknowledgement}

This work was supported by the China Major State Basic Research Development Program (973 Program) No.2012CB316100(2), National Natural Science Foundation of China(NSFC) No.61171064, No. 61321061 and China Scholarship Council.

\end{document}